\documentclass[twoside]{article-hermes}

\usepackage[latin1]{inputenc}
\usepackage[T1]{fontenc}
\usepackage[frenchb]{babel}
\usepackage{fancyhdr}
\usepackage{amsmath}
\usepackage{amsfonts}
\usepackage{amssymb}
\usepackage{latexsym}
\usepackage{theorem}
\usepackage{array}
\usepackage{tabularx}
\usepackage{booktabs}
\usepackage[ruled]{algorithm}
\usepackage{algorithmic}
\usepackage{float}
\usepackage{latexsym}
\usepackage[francais]{layout}
\usepackage[dvips]{graphicx}
\usepackage{epsfig}
\usepackage{textcomp}

\newcommand{\m}[1]{\mathcal{#1}}

\newcommand{\e}{\exists}

\newcommand{\vide}{\emptyset}
\newcommand{\fois}{\bullet}

\newcommand{\ma}{\mathbb{A}}

\providecommand{\cftab}[1]{(\textit{cf.} tableau~\ref{#1})}

\newcommand\ALL{\text{ALL}}
\newcommand\tup[1]{\ensuremath{\mathrm{(}#1\mathrm{)}}}
\newcommand\tupall{\tup{\ALL, \dotsc, \ALL}}
\newcommand\tupvide{\tup{\emptyset, \dotsc, \emptyset}}

\newcommand\OLAP{\textsc{Olap}}
\newcommand\Count{\textsc{Count}}
\newcommand\COUNT{\textsc{Count}}
\newcommand\SUM{\textsc{Sum}}
\newcommand\SGBD{\textsc{Sgbd}}
\newcommand\SQL{\textsc{Sql}}

\newcounter{ex}
\newcounter{ex2}

\newenvironment{example}
{\par \parindent 0cm \refstepcounter{ex} \emph{Exemple \arabic{ex}}}%
{\nopagebreak[4]\par}

\newenvironment{definition}[1]
{\par \parindent 0cm \refstepcounter{ex2} \emph{Définition \arabic{ex2}} \textbf{[#1]} -}%
{\nopagebreak[4]\par}

 {\theoremstyle{plain}
 \theorembodyfont{\upshape}
       \newtheorem{prop}{Proposition}[section]}

 {\theoremstyle{plain}
 \theorembodyfont{\upshape}
       \newtheorem{cor}{Corrolaire}[section]}

 {\theoremstyle{plain}
 \theorembodyfont{\upshape}
       }

{\theoremstyle{plain}
 \theorembodyfont{\upshape}
       }

{\theoremstyle{plain}
 \theorembodyfont{\upshape}
       \newtheorem{theoreme}{Theorème}[section]}

{\theoremstyle{plain}
 \theorembodyfont{\upshape}
       \newtheorem{proof}{Preuve}[section]}

 {\theoremstyle{plain}
 \theorembodyfont{\upshape}
       \newtheorem{remarque}{Remarque}}

 \makeatletter

\newcommand{\algorithmicexit}{\textbf{exit}}
\newcommand{\FORALE}[2][default]{\ALC@it\algorithmicforall\textbf{e}\ #2\ %
\algorithmicdo%
\ALC@com{#1}\begin{ALC@for}}
\newcommand{\RETURN}[2][default]{\ALC@it%
\algorithmicreturn\ #2\ \ALC@com{#1}}
\newcommand{\EXITFOR}{\ALC@it%
\algorithmicexit\ \algorithmicfor}
\newcommand{\EXITLOOP}{\ALC@it%
\algorithmicexit\ \algorithmicloop}
\newcommand{\FORALLOL}[2]{\ALC@it\algorithmicforall%
\ #1\ \algorithmicdo\ #2}
\newcommand{\IFTHENOL}[2]{\ALC@it\algorithmicif%
\ #1\ \algorithmicthen\ #2 \\ }
\newcommand{\IFTHENELSEOL}[3]{\ALC@it\algorithmicif%
\ #1\ \algorithmicthen\ #2\ \algorithmicelse \ #3}
\newcommand{\THENOLO}[1]{\ALC@it \hspace*{2,15cm} \algorithmicelse%
\ #1 \\ }
\newcommand{\THENOLA}[1]{\ALC@it \hspace*{2,65cm} \algorithmicelse%
\ #1 \\ }
\newcommand{\THENOLB}[1]{\ALC@it \hspace*{1,85cm} \algorithmicelse%
\ #1 \\ }
\newcommand{\THENOLC}[1]{\ALC@it \hspace*{1,95cm} \algorithmicelse%
\ #1 \\ }

 \makeatother

\floatname{algorithm}{Alg.}

\renewcommand{\algorithmicif}{\textbf{si}}
\renewcommand{\algorithmicthen}{\textbf{alors}}
\renewcommand{\algorithmicelse}{\textbf{sinon}}
\renewcommand{\algorithmicfor}{\textbf{pour}}
\renewcommand{\algorithmicforall}{\textbf{pour tout}}
\renewcommand{\algorithmicdo}{\textbf{faire}}

\renewcommand{\algorithmicloop}{\textbf{boucle}}

\renewcommand{\algorithmicexit}{\textbf{quitter}}

\begin{document}


\title[Cubes Convexes]{Cubes Convexes}
\journal{Ingénierie des Systèmes d'information -09/2006- Elaboration des entrepôts de données}{-}{-}
\author{Sébastien Nedjar \andauthor Alain Casali \andauthor Rosine Cicchetti \andauthor Lotfi Lakhal }
\address{Laboratoire d'Informatique Fondamentale de Marseille
         (LIF), CNRS UMR 6166 \\
         Universit\'e de la M\'editerran\'ee, Case 901 \\
         163 Avenue de Luminy, 13288 Marseille cedex 9 \\
         \{nedjar, casali, cicchetti, lakhal\}@lif.univ-mrs.fr
}

\resume
{De nombreuses approches ont proposé de
pré-calculer des cubes de données afin de répondre efficacement
aux requêtes \OLAP. La notion de
cube de données a été déclinée de différentes manières : cubes icebergs,
cubes intervallaires ou encore cubes différentiels. Dans cet
article, nous introduisons le concept de cube convexe qui permet
de capturer tous les tuples d'un cube de données satisfaisant une
combinaison de contraintes monotones/antimonotones et peut être représenté de façon très
compacte de manière à optimiser à la fois le temps de calcul et
l'espace de stockage nécessaire. Le cube convexe n'est pas une
structure <<~{\it de plus}~>> à ajouter à la liste des variantes
du cube, mais nous le proposons comme une structure unificatrice
caractérisant, de manière simple, solide et homogène,
les autres types de cubes cités. Enfin, nous proposons une
nouvelle variante : le cube émergent qui met en
évidence les renversements significatifs de tendances. Nous en
proposons une représentation compacte et cohérente avec les
caractérisations précédentes.}
\abstract{In various approaches, data cubes are pre-computed in
order to answer efficiently \OLAP~queries.
The notion of data cube has been declined in various ways: iceberg
cubes, range cubes or differential cubes. In this paper, we
introduce the concept of convex cube which captures all the tuples
of a datacube satisfying a constraint combination. It can be
represented in a very compact way in order to optimize both
computation time and required storage space. The convex cube is
not an additional structure appended to the list of cube variants
but we propose it as a unifying structure that we use to
characterize, in a simple, sound and homogeneous way, the other
quoted types of cubes. Finally, we introduce the concept of
emerging cube which captures the significant trend inversions.
characterizations.}

\motscles{Analyse multi-dimensionnelle, Cubes de données, Cubes convexes, Cubes émergents, Transversaux cubiques}
\keywords{Multidimensional analysis, Datacubes, Convex cubes, Emergent Cubes, Cube Transversals}

\maketitlepage

\section{Introduction et motivations}
En pré-calculant tous les agrégats possibles à différents niveaux
de granularité, les cubes de données \cite{GCBLRVPP97} permettent
une réponse efficace aux requêtes \OLAP~et sont donc un concept
clef pour la gestion des entrepôts. Plus récemment, le calcul de
cubes a été utilisé avec succès pour l'analyse multidimensionnelle
de flots de données \cite{Han05}. Dans ce type d'applications,
d'énormes volumes de données, à un niveau de granularité très fin,
sont générés sous forme de flux continu qu'il est inenvisageable
de balayer plusieurs fois. Or, les utilisateurs de telles
applications dynamiques ont besoin d'une vision plus macroscopique
des données, d'analyser les tendances générales et leurs variations
au cours du temps. Calculer des cubes à partir de flots de données
s'avère
donc très pertinent.

Des travaux de recherche ont proposé différentes variations autour
du concept de cube de données. Par exemple, les cubes icebergs
\cite{BR99} sont des cubes de données partiels qui, à l'instar des
motifs fréquents, ne recèlent que les tendances suffisamment
générales pour être pertinentes, en imposant aux valeurs des différentes mesures d'être
supérieures à des seuils minimaux donnés. Les cubes intervallaires
\cite{CCL03_TSI} peuvent être vus comme une extension des cubes
icebergs dans la mesure où ils permettent à l'utilisateur de se
focaliser sur les tendances qui s'inscrivent dans une <<~{\it
fenêtre}~>> significative ($i.e.$ les mesures sont bornées par des
seuils minimaux et maximaux). Enfin, les nouvelles tendances
apparaissant (ou les tendances avérées disparaissant) lors du
rafraîchissement d'un entrepôt ou dans un flot de données sont
mises en évidence par le calcul de cubes différentiels
\cite{Cas04_dawak}. Ceux-ci peuvent être perçus comme la
différence ensembliste entre deux cubes : l'un stocké dans
l'entrepôt et l'autre calculé à partir des données de
rafraîchissement. Suivant l'ordre des deux opérandes, sont
exhibées les tendances disparaissantes ou apparaissantes.

Souvent ces différents types de cubes, à commencer par le
cube de données originel, n'ont pas été appréhendés comme des concepts mais
comme le résultat de requêtes ou d'algorithmes plus efficaces.

Dans cet article, nous proposons une nouvelle structure
unificatrice qui nous permet de caractériser les différents cubes
évoqués et nous définissons une nouvelle variante : les cubes
émergents. Plus précisément, nos contributions sont les suivantes
:

\begin{itemize}

\item [(i)] nous établissons les fondements d'une nouvelle structure appelée cube convexe qui s'appuie sur l'espace
de recherche du treillis cube. Le cube convexe prend en compte des
combinaisons de contraintes monotones et anti-monotones. Nous
montrons que cette structure est un espace convexe \cite{Vel93} et
qu'elle peut donc être représentée par ses bordures ;

\item [(ii)] grâce à la structure de cube convexe, nous introduisons les définitions
formelles des cubes de données, cubes icebergs, cubes
intervallaires et cubes différentiels ;

\item [(iii)] enfin, en nous inspirant de \cite{DL05} qui définit les
motifs émergents dans un contexte binaire et pour la classification supervisée, 
nous proposons le concept de cube émergent. Celui-ci capture des tendances non
significatives mais qui, lors d'un rafraîchissement, le deviennent
dans des proportions pertinentes pour l'utilisateur. Il permet
aussi, de manière symétrique, d'exhiber des tendances
significatives qui chutent au point de ne plus l'être. Outre les
tendances apparaissant ou disparaissant (capturées par le cube
différentiel), le cube émergent permet au décideur de connaître et
d'analyser des renversements de tendances. La mise en évidence de
changements de tendances est tout aussi intéressante dans
l'analyse de flots de données que dans le contexte de bases
\OLAP~classiques mais elle est plus critique à établir puisqu'en
temps réel.

\end{itemize}

Le reste de l'article est organisé de la manière suivante. Le
paragraphe \ref{sec:travaux_ant} présente notre cadre de travail
en décrivant brièvement l'espace de recherche multidimensionnel
que nous utilisons ensuite : le treillis cube. Au paragraphe
\ref{Section:CubeConvexe}, nous détaillons la structure de cube
convexe. Son utilisation pour caractériser les différents types de
cubes est proposée dans le paragraphe suivant. Enfin, nous
introduisons le concept de cube émergent et le définissons de
manière cohérente avec les précédents.

\section{Concepts de base}
\label{sec:travaux_ant}

Dans ce paragraphe, nous présentons le concept de treillis cube
\cite{CCL03_SIAM} permettant de formaliser les nouvelles
structures introduites par notre proposition.

Tout au long de cet article, nous faisons les hypothèses suivantes
et utilisons les notations introduites. Soit $r$ une relation de
schéma $\m{R}$. Les attributs de $\m{R}$ sont divisés en deux
ensembles ($i$) $\m{D}$ l'ensemble des attributs dimensions (aussi
appelés catégories ou attributs nominaux) qui correspondent aux
critères d'analyse et ($ii$) $\m{M}$ l'ensemble des attributs
mesures.

\subsection{Espace de recherche : treillis cubes} \label{subsection:cl}

L'espace multidimensionnel d'une relation d'attributs catégories
$r$ regroupe toutes les combinaisons valides construites en
considérant l'ensemble des valeurs des attributs de $\m{D}$,
ensemble enrichi de la valeur symbolique $\ALL$.

L'espace multidimensionnel de $r$ est noté et défini comme suit :
\\ $Space(r$) = $(\times_{A~\in~\m{D}} (Dim(A)~\cup$ ALL))
$\cup~\{\tupvide\}$ où $\times$ symbolise le Produit cartésien,
$\tupvide{}$, le majorant universel et $Dim(A)$ la projection de $r$ sur $A$.
Toute combinaison de $Space(r)$ est un tuple et représente un motif multidimensionnel.

L'espace multidimensionnel de $r$ est structuré par la relation de
géneralisation/spé\-cialisation entre tuples, notée $\preceq_g$.
Cet ordre est originellement introduit par T. Mitchell \cite{Mi82}
dans le cadre de l'apprentissage de concepts. Dans un contexte de
gestion de cubes de données, cet ordre a la même sémantique que
celle des opérateurs {\sc Rollup}/{\sc Drilldown} sur le cube
\cite{GCBLRVPP97} et sert de comparateur entre les tuples (cellules) du cube quotient \cite{LPH02}. Soit $u$, $v$ deux tuples de l'espace
multidimensionnel de $r$ :
 $$ u \preceq_g v \Leftrightarrow \left\lbrace
\begin{array}{l} \forall A \in \m{D} \text{ tel que } u[A] \neq
\ALL , \\ \hspace*{1cm} u[A] = v[A] \\
 \text{ou~} v = \tupvide{}
\end{array} \right.  $$
 Si $u \preceq_g v$, nous disons que $u$ est plus général que $v$
 dans $Space(r)$.

\begin{example}
- Considérons la relation \textsc{Document$_1$} ({\it cf} table
\ref{tab:r1}) répertoriant les quantités vendues par Type, par
Ville et par Éditeur.
Dans l'espace multidimensionnel de cette relation, nous avons :
$\tup{\text{Roman, ALL, ALL}}$ $\preceq_g \tup{\text{Roman, Marseille, Gallimard}}$, c'est-à-dire que le tuple 
$\tup{\text{Roman, ALL, ALL}}$ est plus général que $\tup{\text{Roman, Marseille, Gallimard}}$ 
et $\tup{\text{Roman, Marseille, Gallimard}}$ est plus spécifique que 
$\tup{\text{Roman, ALL, ALL}}$. De plus, tout motif
multidimensionnel généralise le tuple $\tup{\vide, \vide, \vide}$
et spécialise le tuple $\tup{\ALL, \ALL, \ALL}$.
\end{example}
\begin{table}[ht!] \begin{center}
\begin{tabular}{ccc|c}
\toprule
   Type      &   Ville    &  Éditeur  & Quantité\\
\midrule
    Roman    &  Marseille & Gallimard &        2\\
    Roman    &  Marseille & Hachette  &        2\\
    Scolaire &    Paris   & Hachette  &        1\\
    Essai    &    Paris   & Hachette  &        6\\
    Scolaire &  Marseille & Hachette  &        1\\
\bottomrule
\end{tabular}
\caption{Relation exemple \textsc{Document$_1$}} \label{tab:r1}
\end{center} \end{table}

Les deux opérateurs de base définis pour la construction de tuples
sont la Somme (notée $+$) et le Produit (noté $\fois$). La somme
de deux tuples retourne le tuple le plus spécifique généralisant
les deux opérandes. Soit \textit{u} et \textit{v} deux tuples de
$Space(r$), $ t = u + v \Leftrightarrow \forall A \in \m{D},$
$$t[A] = \left\lbrace
\begin{array}{l}
u[A]~ \text{si} ~ u[A]~ = ~v[A] \\
 \text{ALL sinon.}\\
\end{array} \right. $$
Nous disons que \textit{t} est la somme des tuples \textit{u} et \textit{v}.

\begin{example} - Dans notre exemple, nous avons $\tup{\text{Roman, Marseille, Gallimard}}
+$ (Roman, Marseille, Hachette) $= \tup{\text{Roman, Marseille, ALL}}$.
Ceci implique que $\tup{\text{Roman, Marseille, ALL}}$ est
construit à partir des tuples $\tup{\text{Roman, Marseille, Gallimard}}$ et 
$\tup{\text{Roman, Marseille, Hachette}}$.
\end{example}

Le produit de deux tuples retourne le tuple le plus général
spécialisant les deux opérandes. Si pour ces deux tuples, il
existe un attribut $A$ prenant des valeurs distinctes et réelles
(i.e. existant dans la relation initiale), alors seul le tuple
$\tupvide{}$ les spécialise (hormis ce tuple, les ensembles
permettant de les construire sont disjoints). Soit $u$ et $v$ deux
tuples de $Space(r)$, alors :   $t = u \fois v \Leftrightarrow$
$$ \left\lbrace
\begin{array}{l}
t = \tupvide{} \text{ si } \exists A \in \m{D} \text{ tel que }
u[A] \neq v[A] \neq \ALL, \\
\text{ sinon } \forall A \in \m{D} \left\lbrace
\begin{array}{l}
t[A] = u[A] \text{ si } v[A] = \ALL \\
t[A] = v[A] \text{ si } u[A] = \ALL.
\end{array} \right.
\end{array} \right. $$

Nous disons que \textit{t} est le produit des tuples \textit{u} et
\textit{v}.
\begin{example} - Nous avons $\tup{\text{Roman, \ALL, \ALL}} \bullet\
\tup{\text{ALL, Marseille, ALL}} = \tup{\text{Roman, Marseille, ALL}}$.
Ainsi, $\tup{\text{Roman, ALL, ALL}}$ et $\tup{\text{ALL,
Marseille, ALL}}$ généralisent (Roman, Marseille, ALL) et ce
dernier tuple participe à la construction de $\tup{\text{Roman,
ALL, ALL}}$ et de $\tup{\text{ALL, Marseille, ALL}}$ (directement
ou non). Les tuples $\tup{\text{Roman, ALL, ALL}}$ et
$\tup{\text{Scolaire, ALL, ALL}}$ n'ont d'autre point commun que le
tuple de valeurs vides (i.e. le tuple $\tup{\vide, \vide,
\vide}$).
\end{example}
 En dotant l'espace multidimensionnel $Space(r)$ de la
relation de généralisation entre tuples et en utilisant les
opérateurs Produit et Somme, nous introduisons une structure
algébrique appelée treillis cube qui fixe un cadre théorique et
général pour l'\textsc{Olap} et la fouille de bases de données multidimensionnelles.
\begin{theoreme} \label{th:TR}
- Soit $r$ une relation d'attributs catégories. L'ensemble ordonné
$CL(r) = \langle Space(r), \preceq_g \rangle$ est un treillis
complet appelé treillis cube dans lequel les opérateurs Meet
($\bigwedge$) et Join ($\bigvee$) sont définis par :
\begin{enumerate}
\item $\forall~ T \subseteq CL(r),~\bigwedge T = +_{t \in T}~ t $
\item $\forall~ T \subseteq CL(r),~\bigvee T = \fois_{t \in T}~ t$
\end{enumerate}
\end{theoreme}
\section{Cubes convexes} \label{Section:CubeConvexe}
Dans ce paragraphe, nous étudions la structure du treillis cube en
présence de conjonctions de contraintes monotones et/ou
antimonotones selon la généralisation. Nous montrons que cette
structure est un espace convexe qu'on appelle cube convexe. Nous
proposons des représentations condensées (avec bordures) du cube
convexe avec un double objectif : définir d'une manière compacte
l'espace de solutions et décider si un tuple $t$ appartient, ou
pas, à cet espace.

Nous prenons en compte les
contraintes monotones et/ou antimonotones les plus couramment
utilisées en fouille de base de données \cite{PH02}. Celles-ci
peuvent porter sur :
\begin{itemize}
 \item des mesures d'intérêts comme la fréquence de motifs, la confiance, la
 corrélation \cite{HaKa01} : dans ce cas, seuls les attributs dimensions de $\m{R}$
 sont nécessaires;
 \item des agrégats selon des attributs mesures $\m{M}$ calculés en utilisant
 des fonctions statistiques additives (\textsc{Count, Sum,
 Min, Max}).
\end{itemize}
\noindent Nous rappelons les définitions de la notion d'espace convexe, des contraintes monotones
et antimonotones selon l'ordre de généralisation $\preceq_g$.
\begin{definition}{Espace Convexe}
Soit $(\m{P}, \leq)$ un ensemble partiellement ordonné, $ \m{C} \subseteq \m{P}$ 
est un espace convexe \cite{Vel93} \textit{si et seulement si} $\forall x,y,z \in \m{P}$ 
tel que $x \leq y \leq z$ et $x,z \in \m{C} \Rightarrow y \in \m{C}$.

Donc $\m{C}$ est borné par deux ensembles : $(i)$ un majorant (ou \textit{<<~Upper set~>>}), 
noté $S$, défini par $S = \max_{\leq}(\m{C})$, $(ii)$ un minorant 
(ou \textit{<<~Lower set~>>}), noté $G$ et défini par $G = \min_{\leq}(\m{C})$.
\end{definition}
\begin{definition}{Contraintes monotones/antimonotones selon la généralisation}
\begin{enumerate}
\item Une contrainte \textit{Const} est dite monotone pour l'ordre de généralisation si et seulement si :  $\forall~t,u \in CL(r) : [t
\preceq_g u$ et $Const(t)] \Rightarrow Const(u)$.
 \item Une contrainte \textit{Const} est dite antimonotone pour l'ordre de généralisation si et
seulement si :  $\forall~t,u \in CL(r) : [t \preceq_g u$ et
$Const(u)] \Rightarrow Const(t)$.
\end{enumerate}
\end{definition}
\noindent \textit{Notations : } nous notons $cmc$ (respectivement
$camc$) une conjonction de contrain\-tes monotones (respectivement
antimonotones) et $chc$ une conjonction hybride de contraintes
(monotones et antimonotones). En reprenant les symboles $S$ et $G$
introduits dans \cite{Mi82} et suivant les cas considérés, les
bornes introduites sont indicées par le type de contrainte
considérée. Par exemple $S_{camc}$ symbolise l'ensemble des tuples
les plus spécifiques vérifiant la conjonction de contraintes
antimonotones. 
\begin{remarque} - Pour éviter les ambiguïtés faites en apprentissage \cite{RK01}, 
il est important de noter que les bornes $S_{chc}$ et $G_{chc}$ ne
sont pas les mêmes que les ensembles $S$ et $G$ définis dans le
cadre de l'espace de versions, car un espace de versions est un
espace convexe, mais tout espace convexe n'est pas un espace de
versions, à cause des contraintes considérées.
\end{remarque}
\begin{example} -  Dans l'espace multidimensionnel exemple de la relation 
\textsc{Document$_1$} \cftab{tab:r1}, nous voulons
connaître tous les tuples dont la somme des valeurs pour
l'attribut mesure {\it Quantité} est supérieure ou égale à 3. La
contrainte <<~\textsc{Sum}(\textit{Quantité}) $\geq$ 3~>> est une
contrainte antimonotone. Si le total des ventes par Type, par
Ville et par Éditeur est supérieur à 3, il l'est {\it a fortiori}
pour un niveau plus agrégé de granularité $e.g.$ par Type et
par Éditeur (toutes villes confondues) ou par Ville (tous types et
éditeurs confondus). De même, si nous voulons connaître tous les
tuples dont la somme des valeurs pour l'attribut \textit{Quantité}
est inférieure ou égale à 6, la contrainte exprimée
<<~\textsc{Sum}(\textit{Quantité}) $\leq$ 6~>> est monotone.
Considérons que le total des ventes par type (les attributs
Ville et Éditeur ont pour valeur $\ALL$) respecte cette contrainte,
la même information observée à un niveau de détail plus fin
satisfait forcément la même condition et donc la somme des ventes
par Type et par Ville est inférieure à 6 de même que la somme
des ventes par Type, par Ville et par Éditeur.
\end{example}

\begin{remarque}
\begin{itemize}
\item Nous supposons par la
suite que le tuple $\tupall{}$ vérifie toujours la conjonction de
contraintes antimonotones et que le tuple $\tupvide{}$ vérifie
toujours la conjonction de contraintes monotones. Avec ces
hypothèses, l'espace des solutions contient au moins un élément
(éventuellement le tuple de valeurs vides).
\item De plus, nous supposons que le tuple $\tupall{}$ ne vérifie jamais la
conjonction de contraintes monotones et que le tuple $\tupvide{}$
ne vérifie jamais la conjonction de contraintes antimonotones, car
sinon, l'espace de solutions est $Space(r)$.
\end{itemize}
\end{remarque}

\begin{theoreme} \label{chap:cl:th:convexe}
- Tout treillis cube avec contraintes monotones et/ou
antimonotones est un espace convexe qu'on appelle cube convexe,
$CC(r)_{const} = \{ t \in CL(r) \mid const(t)\}$, où $const$ peut
être $cmc$, $camc$ ou $chc$ suivant que la combinaison de
contraintes est monotone, antimonotone ou hybride. Son majorant
$S_{const}$ et son minorant $G_{const}$ sont :
\begin{equation*} 
1. \text{ si } const = cmc,\ \left\lbrace
\begin{array}{l}
 G_{cmc} = \min_{\preceq_g}(CC(r)_{cmc}) \\
S_{cmc} = \tupvide{}\\
\end{array} \right. \end{equation*} 

\begin{equation*} 2. \text{ si } const = camc,\
\left\lbrace
\begin{array}{l}
G_{camc} = \tupall{} \\
S_{camc} =
\max_{\preceq_g}(CC(r)_{camc})\\
\end{array} \right. \end{equation*} 

\begin{equation*} 
3.  \text{ si } const = chc,\ \left\lbrace
\begin{array}{l}
G_{chc} = \min_{\preceq_g}(CC(r)_{chc}) \\
S_{chc} = \max_{\preceq_g}(CC(r)_{chc})
\end{array} \right. \end{equation*}
\end{theoreme}

\begin{proof}
\begin{enumerate}
\item Soit $CC(r)_{cmc} = \{ t \in CL(r) \mid
\e~ u \in S_{cmc} \text{ et } \e~ v \in G_{cmc} : t \preceq_{g} u
\text{ et } v \preceq_{g} t\}$. Nous montrons que $CC(r)_{cmc}$
est l'ensemble des tuples satisfaisant la conjonction de
contraintes monotones. Pour les besoins de la démonstration,
notons $Sol_{cmc}$ cet ensemble solution. Soit $t \in
CC(r)_{cmc}$.
\begin{itemize}
\item Par définition, il existe $v \in G_{cmc} \mid cmc(v) \text{
et } v \preceq_{g} t$. Puisque $cmc$ est une contrainte monotone,
nous avons $cmc(t)$. Par conséquent $t \in Sol_{cmc}$. Donc
$CC(r)_{cmc} \subseteq Sol_{cmc}$. (a)
\item Soit $t \in Sol_{cmc}$, il existe forcément $v \in G_{cmc} \mid
v \preceq_{g} t$ car $G_{cmc}$ représente les minimaux vérifiant
$cmc$. De plus la contrainte $\e u \in S_{cmc} \mid t \preceq_{g}
u$ est toujours vérifiée. Donc $t \in CC(r)_{cmc}$ et $Sol_{cmc}
\subseteq CC(r)_{cmc}$. (b)
\end{itemize}
(a) et (b) $\Rightarrow Sol_{cmc} = CC(r)_{cmc}$
\item vrai par application du principe du dualité \cite{FCA}
sur le cube convexe avec une conjonction de contraintes monotones.
 \item  vrai car si $Sol_{chc} = CC(r)_{chc}$, alors $Sol_{chc} = Sol_{cmc} \cap Sol_{camc}$.
L'application des deux caractérisations précédentes permet de
déduire celle pour $chc$.
\end{enumerate}
\end{proof}

Le majorant $S_{const}$ représente les tuples les plus spécifiques
satisfaisant la conjonction de contraintes et le minorant
$G_{const}$ les tuples les plus généraux satisfaisant la
conjonction de contraintes. Donc
 $S_{const}$ et $G_{const}$ permettent d'obtenir des
 représentations condensées du cube convexe en présence d'une
 conjonction de contraintes monotones et/ou antimonotones.

Le corollaire suivant permet la caractérisation des bordures du
cube convexe en présence d'une conjonction  hybride de contraintes
$chc = camc \wedge cmc$ en ne connaissant que ($i$) soit la
bordure maximale pour la contrainte antimonotone ($S_{camc}$) et
la contrainte monotone $cmc$, ($ii$) soit la bordure minimale pour
la contrainte monotone ($G_{cmc}$) et la contrainte antimonotone
$S_{camc}$.

\begin{cor} ~~
\begin{enumerate}
\item \'{E}tant donné $S_{camc}$ et $cmc$, les bordures de l'ensemble
du cube convexe $CC(r)_{chc}$  :
\begin{equation*}
\left\lbrace
\begin{array}{l}
G_{chc} = min_{\preceq_g} (\{t \in CL(r) \mid \exists t' \in
S_{camc}: \\ \hspace*{1cm} t \preceq_g t' \text{ et } cmc(t)\}) \\
 S_{chc} = \{ t \in S_{camc} \mid \exists t' \in G_{chc}: t' \preceq_g t\}
\end{array} \right.
\end{equation*}
\item \'{E}tant donné $G_{cmc}$ et $camc$, une représentation condensée de $CC(r)_{chc}$ est :
\begin{equation*}
\left\lbrace
\begin{array}{l}
S_{chc} = max_{\preceq_g} (\{t \in CL(r) \mid \exists t' \in
G_{cmc}: \\ \hspace*{1cm} t' \preceq_g t \text{ et } camc(t)\}) \\
G_{chc} = \{ t \in G_{cmc} \mid \exists t' \in S_{chc}: t
\preceq_g t'\}.
\end{array} \right.
\end{equation*}
\end{enumerate}
\end{cor}

\begin{example}
- Le tableau \ref{bornes:chc} donne les bornes $S_{camc}$,
$S_{chc}$ $G_{cmc}$ et $G_{chc}$ du cube convexe de la relation
exemple en considérant la contrainte hybride <<~$3 \leq
\textsc{Sum}(\textit{Quantité}) \leq 6$~>>.

La caractérisation du cube convexe comme un espace convexe nous
permet de savoir, en ne connaissant que les bordures du cube
convexe, si un tuple quelconque satisfait ou pas la conjonction de
contraintes. En effet, si un tuple de $Space(r)$ vérifie une
conjonction de contraintes antimonotones alors tout tuple le
généralisant la satisfait aussi. Dualement, si un tuple vérifie
une conjonction de contraintes monotones, alors tous les tuples le
spécialisant satisfont aussi ces contraintes. La représentation
par bordure du cube convexe de la relation \textsc{Document$_1$}
\cftab{bornes:chc} permet de répondre facilement à des requêtes
telles que :
\begin{enumerate}
\item Est ce que le nombre d'achats à Marseille est compris entre 3 et 6
?
\item Est ce que le nombre de livres scolaires vendus à Paris est compris entre 3 et 6
 ?
\item Est ce que le nombre de romans vendus à Aubagne par les éditions Hachette est compris entre 3 et 6
 ?
\end{enumerate}

La réponse à la première question est oui car le tuple
$\tup{\text{ALL, Marseille, ALL}}$, donnant les achats effectués
dans la ville de Marseille tous types et éditeurs confondus,
appartient à la bordure $G_{chc}$. Il en est de même pour la
seconde requête (nombre de livres scolaires vendus à Paris) car le tuple
$\tup{\text{Scolaire, Paris, ALL}}$ spécialise le tuple
$\tup{\text{Scolaire, ALL, ALL}}$ appartenant à $G_{chc}$ et
généralise le tuple $\tup{\text{Scolaire, Paris, Hachette}}$
appartenant à la bordure $S_{chc}$. En revanche, la réponse à la
troisième question est non car le tuple
$\tup{\text{ALL, Paris, Hachette}}$ (tous les livres édités par Hachette achetés à
Paris) ne spécialise aucun tuple de la bordure
$G_{chc}$ et ce même s'il généralise le tuple
$\tup{\text{Scolaire, Paris, Hachette}}$ de la bordure $S_{chc}$.
\end{example}
\begin{table}[H] \begin{center} \label{ex_camc}
\begin{tabular}[h]{c|c}  \toprule
 $S_{camc}$ & \begin{tabular}{c}
$\tup{\text{Roman, Marseille, ALL}}$ \\ $\tup{\text{ALL, Marseille, Hachette}}$ \\
$\tup{\text{Scolaire, Paris, Hachette}}$  \\
          \end{tabular} \\ \hline
$S_{chc}$ & \begin{tabular}{c}
$\tup{\text{Roman, Marseille, ALL}}$ \\ $\tup{\text{ALL, Marseille, Hachette}}$ \\
$\tup{\text{Scolaire, Paris, Hachette}}$  \\
\end{tabular} \\ \hline
 $G_{cmc}$ & \begin{tabular}{c}
$\tup{\text{Roman, ALL, ALL}}$ \\ $\tup{\text{Essai, ALL, ALL}}$ \\
$\tup{\text{Scolaire, ALL, ALL}}$ \\ $\tup{\text{ALL, Marseille, ALL}}$ \\
$\tup{\text{ALL, ALL, Gallimard}}$  \\
               \end{tabular} \\ \hline
 $G_{chc}$ & \begin{tabular}{c}
$\tup{\text{Roman, ALL, ALL}}$ \\ $\tup{\text{Scolaire, ALL, ALL}}$ \\
 $\tup{\text{ALL, Marseille, ALL}}$ \\
         \end{tabular} \\ \bottomrule
\end{tabular}
\caption{\label{bornes:chc} Bornes du cube convexe pour <<~$3 \leq
\textsc{Sum}(\textit{Quantité}) \leq 6$~>>} \end{center}
\end{table}

\section{Formalisation de cubes existants} 
\label{SectionCubesExist}

Dans ce paragraphe, nous passons en revue différentes variantes
des cubes de données et, en utilisant la structure de cube
convexe, nous en proposons une caractérisation à la fois solide et
simple.


\subsection{Cubes de données} \label{subsec:datacub}

Originellement proposé dans \cite{GCBLRVPP97},
le cube de données selon un ensemble de dimensions est présenté comme le
résultat de tous les {\sc Group By} qu'il est possible de formuler
selon une combinaison de ces dimensions. Le résultat de chaque
{\sc Group By} est appelé un cuboïde et l'ensemble de tous les
cuboïdes est structuré au sein d'une relation notée
$Datacube(r)$. Le schéma de cette relation reste le même que celui
de $r$, à savoir $\m{D} \cup \m{M}$ et c'est ce même schéma qui
est utilisé pour tous les cuboïdes (afin de pouvoir en faire
l'union) en mettant en \oe uvre une idée simple : toute dimension
ne participant pas au calcul d'un cuboïde (i.e. ne figurant pas
dans la clause {\sc Group By}) se voit attribuer la valeur ALL.

Pour tout ensemble d'attributs $X \subseteq \m{D}$, un cuboïde du
cube de données, noté $Cuboid(X,f(\{\m{M}|*\}))$, peut être obtenu
comme suit en utilisant une requête {\sc Sql}:
\begin{verbatim}
SELECT [ALL,] X, f({M|*})
    FROM r
    GROUP BY X;
\end{verbatim}

Ainsi, nous obtenons deux requêtes {\sc Sql} pour exprimer le
calcul d'un cube de données : \begin{enumerate} \item soit en
utilisant l'opérateur {\sc Group By Cube}
  (ou {\sc Cube By} selon le \SGBD) :  \begin{verbatim}
 SELECT D, f({M|*})
    FROM r
    GROUP BY CUBE (D);
\end{verbatim}
\item  soit en
 faisant l'union de tous les cuboïdes : \\ $Datacube(r,f(\{\m{M}|*\})) = \bigcup\limits_{X \subseteq \m{D}}
 Cuboid(X,f(\{\m{M}|*\}))$. Cette requête s'exprime comme suit en
 \SQL~:
 \begin{verbatim}
 SELECT  ALL, ..., f({M|*})
    FROM r
 UNION
 SELECT A, ALL, ..., f({M|*})
    FROM r
    GROUP BY A
 UNION
  ...
\end{verbatim}
 \end{enumerate}

\begin{example} -
Dans notre exemple, l'ensemble de toutes les requêtes agrégatives
peut être exprimé en utilisant l'opérateur \textsc{Group By Cube} comme
suit : 
\begin{verbatim}
SELECT Type,Ville,Éditeur,SUM(Quantite)
    FROM Document1
    GROUP BY CUBE Type, Ville, Éditeur;
\end{verbatim}

Cette requête a pour résultat le calcul de $2^3 = 8$ cuboïdes :
$TVE, TE, TV, VE, T, V, E$ et $\vide$ (en considérant les
initiales des attributs). Le cuboïde selon $TVE$ correspond à la
relation initiale elle-même.
\end{example}

Un tuple $t$ appartient au cube de données d'une relation $r$ si
et seulement s'il existe au moins un tuple $t'$ de $r$ qui
spécialise $t$; sinon $t$ ne peut pas être construit. Par
conséquent, quelle que soit la fonction agrégative, les tuples
constituant le cube de données restent invariants, seules les
valeurs calculées par la fonction agrégative changent.

\begin{prop} \label{chap:pb:th:datacube_treillis
cube_contraint}- Soit $r$ une relation projetée sur $\m{D}$,
l'ensemble des tuples ($i.e.$ hormis les valeurs des attributs
mesures) constituant le cube de données de $r$ est un cube convexe
pour la contrainte <<~{\sc Count}(*) $\geq 1$~>> :
$$Datacube(r) = \{ t \in CL(r) \mid t[ Count(*)] \geq 1 \} $$
\end{prop}

\begin{proof}
- D'après la définition d'un cube de données: $t \in datacube(r)
\Leftrightarrow \exists t' \in r\ \mid\ t \preceq_g t'
\Leftrightarrow$ $t[$\Count$(*)] \geq 1$ d'après la définition de
la fonction \Count.
\end{proof}

Puisque la contrainte <<~{\sc Count}(*) $\geq 1$~>> est une
contrainte antimonotone (selon $\preceq_g$), un cube de données
est un cube convexe. En appliquant le théorème
\ref{chap:cl:th:convexe}, nous déduisons que tout cube de données
peut être représenté par deux bordures : la relation $r$ qui est
le majorant et le tuple $\tupall{}$ qui est le minorant. Ainsi,
nous pouvons facilement tester l'appartenance d'un tuple
quelconque $t$ au cube de données de $r$ : il suffit de trouver un
tuple $t' \in r$ qui spécialise $t$.

\begin{example}
- Avec la relation exemple \textsc{Document}$_1$ \cftab{tab:r1}, le
tuple (Roman, Marseille, ALL) appartient au cube de données car il
est spécialisé par le tuple (Roman, Marseille, Gallimard) de la relation
initiale.
\end{example}

Dans les sous-paragraphes suivants, en nous appuyant toujours sur
la structure de cube convexe, nous proposons une caractérisation
des différentes déclinaisons du cube de données.


\subsection{Cubes icebergs}

En  s'inspirant des motifs fréquents, \cite{BR99} introduit les
cubes icebergs qui sont présentés comme des sous-ensembles de
tuples du cube de données satisfaisant, pour les valeurs de la
mesure, une contrainte de seuil minimum. L'objectif sous-jacent
est triple. Il s'agit dans ce cas d'exhiber les tendances
suffisamment générales pour être pertinentes pour le décideur, il
en découle deux intérêts techniques importants : ne pas calculer
ni matérialiser la totalité du cube d'où un gain notable à la fois
de temps d'exécution et d'espace disque.  La requête
\SQL~permettant de calculer un cube iceberg prend la forme
suivante :
\begin{verbatim}
 SELECT D, f({M|*})
    FROM r
    GROUP BY CUBE D
    HAVING f({M|*}) >= MinSeuil;
\end{verbatim}

En nous appuyant sur la définition du cube convexe, nous
formalisons le concept de cube iceberg dans la proposition
suivante :
\begin{prop}
- La contrainte <<~\verb+f({M|*}) >= MinSeuil+~>> étant une
contrainte antimonotone \cite{PH02}, le cube iceberg est un cube
convexe caractérisé comme suit :   
$$CubeIceberg(r) = \{ t \in CL(r) \mid t[f(\{M|*\})] \geq  MinSeuil \}.$$
\end{prop}

\begin{example}
- Avec la relation exemple \textsc{Document}$_1$ \cftab{tab:r1}, la
bordure $S$ relative à la contrainte <<~\SUM(Quantité) $\geq$ 3~>>
est composée des trois tuples suivants : $ \{$ (Roman, Marseille,
ALL), (ALL, Marseille, Hachette), (Roman, Paris, Hachette) $\}$. La
bordure $G$ est uniquement composée du tuple ne contenant que des
valeurs $\ALL$.
\end{example}

\subsection{Cubes intervallaires}

Le cube intervallaire ne contient que les tuples du cube de
données pour lesquels les valeurs de la mesure sont comprises dans
un intervalle donné. La requête \SQL~permettant de calculer un tel
cube est la suivante :

\begin{verbatim}
 SELECT D, f({M|*})
 FROM r
 GROUP BY CUBE D
 HAVING f({M|*}) BETWEEN MinSeuil AND MaxSeuil;
\end{verbatim}

Dans le cadre de travail établi, la caractérisation du cube
intervallaire est la suivante :

 \begin{prop} - La contrainte <<~\verb+f({M|*}) >= MinSeuil+~>>
étant une contrainte antimonotone \cite{PH02} et la contrainte 
<<~\verb+f({M|*}) <= MaxSeuil+~>> étant une contrainte monotone, le
cube intervallaire est un cube convexe. Ainsi, nous avons :
$$CubeIntervalaire(r) = \{ t \in CL(r) \mid MaxSeuil \geq
t[f(\{M|*\})] \geq MinSeuil \}.$$
\end{prop}

\begin{example}
- Avec la relation exemple \textsc{Document}$_1$ \cftab{tab:r1}, les
bordures $S$ et $G$ relatives à la contrainte <<~\SUM(Quantité)
$\in [3,6]$~>> sont données dans le tableau \ref{bornes:interval}.
\end{example}
\begin{table}[H] \begin{center}
\begin{tabular}[h]{c|c}  \toprule
 $S$ & \begin{tabular}{c}
$\tup{\text{Roman, Marseille, ALL}}$ \\ $\tup{\text{ALL, Marseille, Hachette}}$ \\
$\tup{\text{Scolaire, Paris, Hachette}}$  \\
          \end{tabular} \\ \hline
 $G$ & \begin{tabular}{c}
$\tup{\text{Roman, ALL, ALL}}$ \\ $\tup{\text{Scolaire, ALL, ALL}}$ \\
 $\tup{\text{ALL, Marseille, ALL}}$\\
         \end{tabular} \\ \bottomrule
\end{tabular}
\caption{\label{bornes:interval} Bornes du cube intervalaire pour
la contrainte  <<~$\textsc{Sum}(\textit{Quantité}) \in [3,6]$~>>}
\end{center}
\end{table}

\subsection{Cubes différentiels} \label{subsec:diff_cube}

Les cubes différentiels sont le résultat de la différence entre
les cubes de deux relations $r_1$ et $r_2$. Ils mettent en
évidence des tuples pertinents dans un des cubes et inexistants
dans l'autre. Leur intérêt est donc de pouvoir comparer des
tendances entre deux jeux de données. Par exemple, dans une
application distribuée où deux relations rassemblent des données
collectées dans des zones géographiques différentes, le cube
différentiel exhibe des tendances significatives dans une zone
mais inexistantes dans l'autre. Si l'on considère non plus la
dimension géographique mais la dimension temporelle, les cubes
différentiels permettent d'isoler des tendances fréquentes à un
instant et qui disparaissent ou des tendances inexistantes qui
apparaissent de manière significative. Considérons que la relation
originelle est $r_1$ et que les tuples alimentant l'entrepôt lors
d'un rafraîchissement sont stockés dans $r_2$, le cube
différentiel peut être obtenu par la requête \SQL~suivante :

\begin{verbatim}
 SELECT D, f({M|*})
    FROM r2
    GROUP BY CUBE D
    HAVING f({M|*}) >= MinSeuil
 MINUS
 SELECT D, f({M|*})
    FROM r1
    GROUP BY CUBE D;
\end{verbatim}

De manière cohérente avec les types de cubes précédents, nous
proposons une caractérisation des cubes différentiels.

\begin{prop} - La contrainte <<~\verb+f({M|*}) >= MinSeuil+~>> étant une
contrainte antimonotone \cite{PH02} et la contrainte <<~{\it
n'appartient pas au cube de données de la relation $r_1$}~>> étant
une contrainte monotone ($cf.$ paragraphe \ref{subsec:datacub} et
application du principe de dualité \cite{FCA}), le cube
différentiel est un cube convexe. Ainsi, nous avons : 
\begin{displaymath}
\begin{array}{lcl}
CubeDifferentiel(r_2,r_1) &=&\{ t \in CL(r_1 \cup r_2) \mid
t[f(\{M|*\})] \geq MinSeuil \\ &&\text{ et } \nexists t' \in r_2 \mid t
\preceq_g t' \}. 
\end{array}
\end{displaymath}

\end{prop}

\begin{example}
- Soit la relation {\sc Document$_2$} suivante :
\begin{table}[H] \begin{center}
\begin{tabular}{c|ccc|c}
\toprule RowId&    Type&      Ville&       Éditeur&   Quantité\\
\midrule
    1&     Scolaire&  Marseille & Gallimard &      3\\
    2&     Scolaire&  Paris     & Hachette  &      3\\
    3&     Scolaire&  Marseille & Hachette  &      1\\
    4&     Roman   &  Marseille & Gallimard &      3\\
    5&     Essai   &  Paris     & Hachette  &      2\\
    6&     Essai   &  Paris     & Gallimard &      2\\
    7&     Essai   &  Marseille & Hachette  &      1\\
\bottomrule
\end{tabular}
\caption{Relation exemple \textsc{Document$_2$}} \label{tab:r2}
\end{center} \end{table}

Les bordures $S$ et $G$ du cube différentiel entre les relations
\textsc{Document$_2$} et \textsc{Document$_1$}, pour le seuil $MinSeuil
= 1/15$ et la fonction agrégative {\sc Sum}, sont données dans le
tableau \ref{bornes:differentiel}.

\begin{table}[H] \begin{center}
\begin{tabular}[h]{c|c}  \toprule
 $S$ & \begin{tabular}{c}
$\tup{\text{Scolaire, Marseille, Gallimard}}$ \\
 $\tup{\text{Scolaire, Marseille, Hachette}}$ \\
$\tup{\text{Essai, Paris, Gallimard}}$ \\
          \end{tabular} \\ \hline
 $G$ & \begin{tabular}{c}
$\tup{\text{Essai, ALL, Gallimard}}$ \\ $\tup{\text{Scolaire, Marseille, ALL}}$ \\
 $\tup{\text{Scolaire, ALL, Gallimard}}$ \\   $\tup{\text{ALL, Paris, Galimard}}$ \\
         \end{tabular} \\ \bottomrule
\end{tabular}
\caption{\label{bornes:differentiel} Bornes du cube différentiel}
\end{center}
\end{table}

\end{example}

\begin{remarque}
- Dans un souci d'homogénéité, nous terminons ce paragraphe en
donnant la requête \SQL~générique qui correspond au calcul du cube
convexe :

\begin{verbatim}
 SELECT D, f({M|*})
 FROM r
 GROUP BY CUBE D
 [HAVING condition(s) anti-monotone(s)
        AND Condition(s) monotone(s)];
\end{verbatim}
\end{remarque}


\section{Cubes émergents} \label{SectionCubesEmerg}

Dans ce paragraphe, nous introduisons le concept de cube émergent.
De tels cubes exhibent des tendances non pertinentes pour
l'utilisateur (parce qu'en deçà d'un seuil) mais qui le deviennent
ou au contraire des tendances significatives qui s'atténuent sans
forcément disparaître. Les cubes émergents permettent donc
d'élargir les résultats des cubes différentiels en affinant les
comparaisons entre deux cubes. Ils sont tout aussi intéressants
dans un contexte \OLAP~que pour l'analyse de flots de données car
ils mettent en évidence des renversements de tendances.

 Dans la suite, nous considérons
uniquement les fonctions agrégatives \COUNT~et \SUM. Pour
conserver la propriété d'antimonotonie de \SUM, nous supposons que
toutes les valeurs prises par la mesure sont strictement positives
et introduisons une version relative de ces fonctions.

\begin{definition}{Fonction agrégative relative}
Soit $r$ une relation, $t \in CL(r)$ un tuple, et $f \in \{$\SUM,
\COUNT$\}$ une fonction agrégative. On appelle $f_{rel}(.,r)$ la
fonction agrégative relative de la fonction $f$ pour la relation
$r$. $f_{rel}(t,r)$ est le ratio entre la valeur de $f$ pour le
tuple $t$ et la valeur de $f$ appliquée sur toute la relation $r$
(donc pour le tuple $\tupall{}$).
$$ f_{rel}(t,r) = \frac{f(t,r)}{f(\tupall,r)}$$
\end{definition}

Par exemple, la fonction \COUNT$_{rel}(t,r)$ correspond simplement
à $Freq(t,r)$ (la fréquence d'apparition du motif
multidimensionnel $t$ dans la relation $r$).

\begin{remarque}
- Étant donné que $f$ est additive et que les valeurs prises par
la mesure sont strictement positives, on a $0 < f(t, r) < f(\tupall, r)$
et par conséquent $0 < f_{rel}(t,r) < 1$
\end{remarque}

Disposant de deux jeux de données unicompatibles $r_1$ et $r_2$, nous nous intéressons 
aux tendances <<~non significatives~>> dans $r_1$ mais qui le deviennent dans $r_2$
dans des proportions pertinentes pour l'utilisateur.
Les entrepôts aussi bien que les flots de données incluant nécessairement 
la dimension chronologique, $r_1$ et $r_2$ peuvent être typiquement vus comme
des ensembles de tuples pourvus d'estampilles temporelles différentes.
Ainsi, $r_1$ peut déjà être stockée dans l'entrepôt et $r_2$ rassembler les nouveaux
tuples à insérer lors d'un rafraîchissement.
$r_1$ et $r_2$ peuvent aussi correspondre à des ensembles de tuples
collectés lors de deux intervalles de temps. 
Même si la dimension temps est prépondérante dans notre contexte de travail, 
elle n'est pas nécessairement le critère de comparaison entre $r_1$ et $r_2$. 
Ainsi, dans une application distribuée, les tuples de $r_1$ et $r_2$
peuvent être collectés sur deux sites différents. Les tuples émergents de $r_1$ vers $r_2$
peuvent être simplement caractérisés par deux contraintes de seuil.

\begin{definition}{Tuple émergent}
Un tuple $t\in CL(r_2\cup r_1)$ est dit émergent de $r_1$ vers $r_2$ si et
seulement s'il satisfait les deux contraintes suivantes :
\begin{itemize}
\item[$(C_1)$] $f_{rel}(t,r_1) \leq MinSeuil_1 $
\item[$(C_2)$] $f_{rel}(t,r_2) \geq MinSeuil_2 $
\end{itemize}
où $MinSeuil_1$ et $MinSeuil_2 \in\ ]0,1[$
\end{definition}

\begin{example}
- Soit les seuils $MinSeuil_1 = 1/3$ pour la relation {\sc
Document$_1$} \cftab{tab:r1} et $MinSeuil_2 = 1/5$ relatif à la
relation {\sc Document$_2$} \cftab{tab:r2}, le tuple $t_1 =
\tup{\text{Essai, Paris, ALL}}$ est émergent de {\sc Document$_1$} vers {\sc Document$_2$} 
car $SUM_{rel}(t_1, r_1) = 1/12$ et $SUM_{rel}(t_1, r_2) = 4/15$. 
Par contre, le tuple $t_2 =$(Essai, Marseille, ALL) ne l'est pas car 
$SUM_{rel}(t_2, r_2) = 1/15$.
\end{example}

\begin{definition}{Cube émergent}
Nous appelons cube émergent l'ensemble de tous les tuples de
$CL(r_2\cup r_1)$ émergents de $r_1$ vers $r_2$. \\
Le cube émergent, noté $CubeEmergent(r_2,r_1)$, est un cube
convexe avec la contrainte hybride $(C_1 \wedge C_2)$ \emph{<<~$t$
est émergent de $r_1$ vers $r_2$~>>}. Il est donc défini ainsi :
 $CubeEmergent(r_2,r_1) =
\{t \in CL(r_1\cup{}r_2) \mid C_1(t) \wedge C_2(t) \},$ avec
$C_1(t) = f_{rel}(t,r_1) < MinSeuil_1 $ et $C_2(t) =
f_{rel}(t,r_2) \geq MinSeuil_2$.
\end{definition}

$CubeEmergent(r_2,r_1)$ est un cube convexe avec une conjonction
de contraintes monotone ($C_1$) et antimonotone ($C_2$). On peut
donc utiliser ses bordures ({\it cf} théorème
\ref{chap:cl:th:convexe}) pour répondre efficacement à la question
\emph{<<~un tuple $t$ est il émergent?~>>}.

\begin{definition}{Taux d'émergence}
Soit $t \in CL(r_1 \cup r_2)$ un tuple
et $f$ une fonction additive (appliquée sur les valeurs toutes positives de la mesure).
Nous notons $TE(t)$ le taux d'émergence de $t$ entre $r_1$ et $r_2$ et nous
le définissons ainsi :
\begin{displaymath}
TE(t) =
\left
\lbrace 
\begin{array}{llll}
0 & si f_{rel}(t,r_1) = 0& ~et~ &f_{rel}(t,r_2) = 0\\
\infty & si f_{rel}(t,r_1) = 0& ~et~ &f_{rel}(t,r_2) \neq 0\\
\dfrac{f_{rel}(t,r_2)}{f_{rel}(t,r_1)}&sinon
\end{array} \right.
\end{displaymath}
\end{definition}

\`A l'instar de la mesure de corrélation \cite{HaKa01}, lorsqu'un
tuple a un taux d'émergence strictement supérieur à 1, il est
positivement émergent, sinon il est négativement émergent.

\begin{example}
- Le tableau \ref{tab:cube_emergent}  présente l'ensemble des
tuples émergents de \textsc{Document}$_1$ vers \textsc{Document}$_2$ en considérant
la mesure Quantité et les seuils $MinSeuil_1 = 1/3$ et
$MinSeuil_2 = 1/5$.

\begin{table}[ht] \begin{center}
\begin{tabular}{ccc|c}
\toprule     Type&      Ville&       Éditeur&         TE({\it Quantité})\\
\midrule
           Scolaire &  Marseille & Gallimard &   $\infty$\\
           Scolaire &  Marseille & ALL       &   $\infty$\\
           Scolaire &  ALL       & Gallimard &   $\infty$\\
           ALL      &  Marseille & Gallimard &   2.4 \\
           ALL      &  ALL       & Gallimard &   2.4 \\
           Roman    &  Marseille & Gallimard &   1.2 \\
           Roman    &  ALL       & Gallimard &   1.2\\
           Essai    &  Paris     & ALL       &   3.2\\
           Essai    &  ALL       & ALL       &   2\\
\bottomrule
\end{tabular}
\caption{Ensemble des tuples émergents de \textsc{Document}$_1$ vers
\textsc{Document}$_2$} \label{tab:cube_emergent}
\end{center} \end{table}
\end{example}

On observe que le taux d'émergence, quand il est supérieur à 1 (donc positivement émergent),
permet de caractériser les tendances significatives dans $r_2$ et
qui ne sont pas aussi marquées dans $r_1$. Quand ce taux est
inférieur à 1, il met en évidence les tendances immergentes, {\it
i.e.} significatives dans $r_1$ et peu présentes ou inexistantes
dans $r_2$.

\begin{example} - Dans les deux relations données en exemple, on a $TE($Scolaire, ALL, ALL$) = 2.5$.
\'Evidemment, plus le taux d'émergence est élevé, plus la tendance
est forte. Ainsi, le tuple ci-avant indique un bond pour la vente
de livres scolaires entre \textsc{Document}$_1$ et \textsc{Document}$_2$.
\end{example}

\begin{prop}
- Soit $MinRatio = \frac{MinSeuil_2}{MinSeuil_1}$, $\forall t \in
CubeEmergent(r_2,r_1)$, on a $TE(t) \geq MinRatio$.
\end{prop}

\begin{proof} - $f_{rel}(t,r_1) \leq MinSeuil_1 $  $\Rightarrow \dfrac{1}{f_{rel}(t,r_1)} \geq \dfrac{1}{MinSeuil_1}$, or $f_{rel}(t,r_2) \geq MinSeuil_2 $ \\
  $\Rightarrow \dfrac{f_{rel}(t,r_2)}{f_{rel}(t,r_1)} \geq \dfrac{MinSeuil_2}{MinSeuil_1}$ \\
  $\Rightarrow TE(t) \geq MinRatio$ $\Box$
\end{proof}

Tous les tuples émergents de notre exemple, ont un taux d'émergence supérieur à
$3/5$. Ceux qui ont un taux d'émergence strictement supérieur à 1
sont positivement émergents, les autres sont donc immergents.

Le cube émergent étant un cube convexe, il peut être représenté
par ses bordures et donc sans avoir à calculer ni matérialiser les
deux cubes comparés. Cette capacité est particulièrement
attrayante car elle permet d'isoler les renversements de tendances
extrêmement rapidement et à moindre coût.

\subsection{Transversaux cubiques} \label{subsection:ctr}

 Nous présentons le concept de transversaux cubiques \cite{CCL03_KDD} qui est
un cas particulier des transversaux d'un hypergraphe \cite{Ber89,EG95,GKMT97} dans 
le contexte du treillis cube.

\begin{definition}{Transversal cubique} \label{def:rel_trans} Soit
$T$ un ensemble de tuples ($T \subseteq CL(r)$) et soit $t \in T$
un tuple, $t$ est un transversal cubique de $T$ sur $CL(r)$ 
si et seulement si $t$ est un transversal cubique et 
$\forall t' \in T, t' \text{ est un transversal cubique et } t' \preceq_g t \Rightarrow t = t'$. 
Les minimaux transversaux cubiques de $T$ sont notés $cTr(T)$ et définis
comme suit :
$$cTr(T) = \min_{\preceq_g}(\{t \in CL(r) \mid
\forall t' \in r, t+t' \neq  \tupall{} \})$$

Soit $\ma$ une anti-chaîne de $CL(r)$ (tous les tuples de
 $\ma$ sont incomparables selon $\preceq_g$), l'ensemble des minimaux transversaux cubiques de
  $T$ peut être contraint en utilisant $\ma$. La nouvelle définition associée est la suivante :
 $$cTr(T,\ma) = \{t \in cTr(r) \mid \e u \in \ma : t \preceq_g u\}$$
\end{definition}

\begin{example}
- Avec la relation exemple {\sc Document$_1$}, nous avons le résultat
suivant : $cTr({\textsc{Document}_1}) = \{$ (Roman, Paris, ALL),
(Essai, ALL, Gallimard), (Scolaire, Marseille, ALL), (Scolaire, ALL, Gallimard), 
(ALL, Paris, Gallimard) $\}$.
\end{example}
Dans le prochain paragraphe, nous
montrons que l'on peut utiliser ce concept pour donner une nouvelle 
formulation des bordures du cube émergent.

\subsection{Calcul efficace des bordures} Pour calculer les bordures de
l'ensemble $CubeEmergent(r_2,r_1)$, nous reformulons les
contraintes de manière à tirer profit d'algorithmes existants qui
ont fait preuve de leur efficacité : ($i$) Max-Miner \cite{Bay98}
et GenMax \cite{GZ01} pour le calcul des maximaux cubiques, et
($ii$) Trans \cite{EG95}, CTR \cite{CCL03_KDD}, MCTR 
\cite{Cas04_dawak}  et \cite{GKMT97} pour le calcul des minimaux transversaux
cubiques. Nous ramenons la contrainte \emph{<<~t est un
tuple émergent~>>} 
à la recherche de maximaux cubiques fréquents et de minimaux
cubiques transversaux.

Il a été démontré que la contrainte $(C_1)$ est une contrainte
monotone et $(C_2)$ est une contrainte anti-monotone pour l'ordre
de généralisation.

L'ensemble des tuples émergents peut être représenté via deux
bordures : $S$ qui contient l'ensemble des tuple maximaux
émergents et $G$ englobant l'ensemble des tuples minimaux
émergents.

$$  \left\lbrace
\begin{array}{l}
G = \min_{\preceq_g}(\{t \in CL(r) \mid C_1(t) \wedge C_2(t)\})\\
\\
S = \max_{\preceq_g}(\{t \in CL(r) \mid C_1(t) \wedge C_2(t)\})
\end{array} \right.
$$

\begin{prop}
- Soit $M_1$ et $M_2$ les tuples maximaux fréquents des relations
$r_1$ et $r_2$ : \\ $M_1 = \max_{\preceq_g} ( \{ t\in CL(r_1)\ :\
f_{rel}(t,r_1) \geq MinSeuil_1 \}  )$\\
$M_2 = \max_{\preceq_g} ( \{ t \in CL(r_2)\ :\ f_{rel}(t,r_2) \geq
MinSeuil_2 \}  )$

\noindent Nous pouvons alors caractériser les bordures $S$ et $G$
de l'ensemble des tuples émergents comme suit : \begin{enumerate}
\item $G = cTr({M_1},M_2)$ sur $CL(r_1 \cup r_2)$,
\item $S = \{ t \in M_2\ : \exists u \in G\ : u \preceq_g t \}$.
\end{enumerate}
\end{prop}

\begin{proof}~~
\begin{enumerate}
\item
$t \in G \Leftrightarrow t \in \min_{\preceq_g} (\{ u \in CL(r_1 \cup r_2)\ : f_{rel}(u,r_1) \leq MinSeuil_1 \text{ et } f_{rel}(u,r_2) \geq MinSeuil_2 \})$ \\
$\Leftrightarrow t \in \min_{\preceq_g} (\{ u \in CL(r_1 \cup r_2)\ : f_{rel}(u,r_1) \leq MinSeuil_1 \})$ et
   $\exists v  \in M_2\ : t \preceq_g v$ \\
$\Leftrightarrow t \in \min_{\preceq_g} (\{ u \in CL(r_1 \cup r_2)\ :  \nexists v \in M_1\ : u \preceq_g v \}) \text{ et }   \exists v  \in M_2\ : t \preceq_g v$ \\
$\Leftrightarrow t \in cTr( {M_1} )$ et   $\exists v  \in M_2\ : t \preceq_g v$ \\
$\Leftrightarrow t \in cTr( {M_1}, M_2 )$

\item Vrai car $CubeEmergent(r_2,r_1)$ est un espace convexe; par conséquent tout tuple $t \in S$
spécialise au moins un tuple $v \in G$. $\Box$
\end{enumerate}
\end{proof}

\begin{remarque} - Cette nouvelle caractérisation utilisant les
minimaux transversaux cubiques s'applique, dans un contexte
binaire, aussi bien pour le calcul des motifs émergents
\cite{DL05}, que pour celui des bordures des motifs contraints
selon une conjonction hybride \cite{RK01} en utilisant le concept
classique de transversal \cite{EG95,Ber89}.
\end{remarque}

\begin{table}[H] \begin{center}
\begin{tabular}[h]{c|c}  \toprule
 $M_1$ & \begin{tabular}{c}
$\tup{\text{Roman, Marseille, Gallimard}}$ \\
 $\tup{\text{Scolaire, Paris, Hachette}}$ \\
          \end{tabular} \\ \bottomrule
\end{tabular}
\caption{Ensemble $M_1$ des tuples maximaux fréquents de \textsc{Document}$_1$} \label{tab:max_r1}
\end{center}
\end{table}

\begin{table}[H] \begin{center}
\begin{tabular}[h]{c|c}  \toprule
 $M_2$ & \begin{tabular}{c}
$\tup{\text{Scolaire, Marseille, Gallimard}}$ \\
 $\tup{\text{Scolaire, Paris, Hachette}}$ \\
$\tup{\text{Roman, Marseille, Gallimard}}$ \\
$\tup{\text{Essai, Paris, ALL}}$ \\
          \end{tabular} \\ \bottomrule
\end{tabular}
\caption{Ensemble $M_2$ des tuples maximaux fréquents de \textsc{Document}$_2$} \label{tab:max_r2}
\end{center}
\end{table}

\begin{example}
- Considérons les relations $r_1 = \textsc{Document}_1$ et
$r_2 = \textsc{Document}_2$. Les ensembles $M_1$ et $M_2$ sont donnés
dans les tableaux \ref{tab:max_r1} et \ref{tab:max_r2}.

Les bordures de $CubeEmergent(r_2, r_1)$ avec $MinSeuil_1 = 1/3$ et
$MinSeuil_2 = 1/5$ sont présentées dans les tableaux
\ref{tab:bordure_U} et \ref{tab:bordure_L}. Grâce à ces deux
bordures, nous pouvons affirmer que les livres scolaires se sont mieux
vendus à Marseille dans la deuxième relation car
\tup{\text{Scolaire, Marseille, ALL}} généralise le tuple
\tup{\text{Scolaire, Marseille, Hachette}} appartenant à $S$ et il est
généralisé par (Scolaire, Marseille, ALL) appartenant à $G$. En
revanche, on ne peut rien affirmer pour le tuple
$\tup{\text{Scolaire, ALL, Gallimard}}$ car il ne spécialise aucun tuple
de $G$.
\end{example}

\begin{table}[H] \begin{center}
\begin{tabular}[h]{c|c}  \toprule
 $S$ & \begin{tabular}{c}
$\tup{\text{Scolaire, Marseille, Gallimard}}$ \\
$\tup{\text{Roman, Marseille, Gallimard}}$ \\
$\tup{\text{Essai, Paris, ALL}}$ \\
          \end{tabular} \\ \bottomrule
\end{tabular}
\caption{Majorant $S$ de $CubeEmergent($\textsc{Document}$_2$, \textsc{Document}$_1)$} \label{tab:bordure_U}
\end{center}
\end{table}

\begin{table}[H] \begin{center}
\begin{tabular}[h]{c|c}  \toprule
 $G$ & \begin{tabular}{c}
$\tup{\text{ALL, ALL, Gallimard}}$ \\
$\tup{\text{Scolaire, Marseille, ALL}}$ \\
$\tup{\text{Essai, ALL, ALL}}$ \\
          \end{tabular} \\ \bottomrule
\end{tabular}
\caption{Minorant $G$ de $CubeEmergent($\textsc{Document}$_2$, \textsc{Document}$_1)$} \label{tab:bordure_L}
\end{center}
\end{table}

\section{Conclusion}

Nous avons, dans cet article, passé en revue différentes
déclinaisons du concept de cube de données en leur ajoutant une
nouvelle variation : les cubes émergents. En mettant en évidence
les renversements de tendances ou, plus précisément, leurs
évolutions dans des proportions significatives pour le décideur,
les cubes émergents apportent de nouvelles connaissances
particulièrement pertinentes pour comparer deux jeux de données.
Leur représentation compacte et leur calcul efficace en font des
candidats de choix pour toutes les applications d'analyse
multidimensionnelle de flots de données. En effet, les
utilisateurs de telles applications dynamiques cherchent
expressément à connaître toute évolution de tendances pour être à
même d'y réagir en temps réel.

Outre la proposition de ce nouveau type de cube, nous avons défini
une structure unificatrice, le cube convexe, qui est un cadre
formel et générique permettant de caractériser, de manière simple
et solide, différentes variantes de cubes de données, trop souvent
perçus comme les résultats de requêtes ou d'algorithmes et non
comme des concepts. Nous nous sommes attachés à mettre en regard
ces deux perceptions. Il résulte de ce travail une caractérisation
homogène des divers types de cubes examinés, une classification
qui se veut didactique pour que l'utilisateur choisisse la
variante de cube adaptée à ses besoins mais surtout la possibilité
d'une représentation compacte, solidement établie pour la
structure générique du cube convexe et que nous montrons
applicable à ses déclinaisons spécifiques que sont le cube
iceberg, le cube intervallaire, le cube différentiel, le cube
émergent et le cube de données lui-même. Quelle que soit la variante de
cube considérée, la représentation par bordure obtenue est, à
l'heure actuelle, la meilleure, {\it i.e.} la plus petite possible
et donc, dans des contextes où ces considérations sont cruciales,
la moins coûteuse à calculer (d'autant qu'il existe des
algorithmes ayant prouvé leur efficacité) et la moins coûteuse à
matérialiser.


\bibliography{biblio}
\end{document}